\documentclass{llncs}

\usepackage{hyperref}
\usepackage{graphicx}
\usepackage{amsmath}
\usepackage{amsfonts}
\usepackage{amssymb}
\usepackage{comment}

\let\doendproof\endproof
\renewcommand\endproof{\hfill$\qed$\doendproof}

\newcommand{\dirN}{\textrm{\textsc{N}}}
\newcommand{\dirE}{\textrm{\textsc{E}}}

\newcommand{\dirW}{\textrm{\textsc{W}}}

\begin{document}

\title{Size-Separable Tile Self-Assembly:\\A Tight Bound for Temperature-1 Mismatch-Free Systems}

\titlerunning{Size-Separable Tile Self-Assembly}

\author{Andrew Winslow}

\authorrunning{Andrew Winslow}
\tocauthor{Andrew Winslow}

\institute{
Tufts University, Medford, MA 02155, USA.\\
\email{awinslow@cs.tufts.edu}
}

\maketitle

\begin{abstract}
We introduce a new property of tile self-assembly systems that we call \emph{size-separability}.
A system is size-separable if every terminal assembly is a constant factor larger than any intermediate assembly.
Size-separability is motivated by the practical problem of filtering completed assemblies from a variety of incomplete ``garbage'' assemblies using gel electrophoresis or other mass-based filtering techniques.

Here we prove that any system without cooperative bonding assembling a unique mismatch-free terminal assembly can be used to construct a size-separable system uniquely assembling the same shape.
The proof achieves optimal scale factor and temperature for the size-separable system.
As part of the proof, we obtain two results of independent interest on mismatch-free temperature-1 two-handed systems.
\keywords{2HAM, hierarchical, aTAM, glues, gel electrophoresis}
\end{abstract}

\section{Introduction}

The study of theoretical tile self-assembly was initiated by the Ph.D. thesis of Erik Winfree~\cite{Winfree-1998a}.
He proved that systems of passive square particles (called \emph{tiles}) that attach according to matching bonds (called \emph{glues}) are capable of universal computation and efficient assembly of shapes such as squares.
Soloveichik and Winfree~\cite{Soloveichik-2007a} later proved that these systems are capable of efficient assembly of any shape, allowing for an arbitrary scaling of the shape, used to embed a roving Turing machine.
In this original \emph{abstract Tile Assembly Model (aTAM)}, tiles attach singly to a growing seed assembly.

An alternative model, called the \emph{two-handed assembly model (2HAM)}~\cite{Abel-2010a,Cannon-2013a,Demaine-2008b,Doty-2010b}, \emph{hierarchical tile assembly model}~\cite{Chen-2012a,Padilla-2013a}, or \emph{polyomino tile assembly model}~\cite{Luhrs-2009a,Luhrs-2010a}, allows ``seedless'' assembly, where tiles can attach spontaneously to form large assemblies that may attach to each other.
This seedless assembly was proved by Cannon et al.~\cite{Cannon-2013a} to be capable of simulating any seeded assembly process, while also achieving more efficient assembly of some classes of shapes.

A generalization of the 2HAM called the \emph{staged tile assembly model} introduced by Demaine et al.~\cite{Demaine-2008b} utilizes sequences of \emph{mixings}, where each mixing combines a set of \emph{input assemblies} using a 2HAM assembly process.
The products of the mixing are the \emph{terminal assemblies} that cannot combine with any other assembly produced during the assembly process (called a \emph{producible assembly}).
This set of terminal assemblies can then be used as input assemblies in another mixing, combined with the sets of terminal assemblies from other mixings.

After a presentation by the author of work~\cite{Winslow-2013a} on the staged self-assembly model at DNA 19, Erik Winfree commented that the staged tile assembly model has a unrealistic assumption: at the end of each mixing process, all producible but non-terminal assemblies are removed from the mixing.
A similar assumption is made in the 2HAM model, where only the terminal assemblies are considered to be ``produced'' by the system.

Ignoring large producible assemblies is done to simplify the model definition, but allows unrealistic scenarios where ``nearly terminal'' systems differing from some terminal assembly by a small number of tiles are presumed to be eliminated or otherwise removed at the end of the assembly process.
While filtering techniques, including well-known gel electrophoresis, may be employed to obtain filtering of particles at the nanoscale, such techniques generally lack the resolution to distinguish between macromolecules that differ in size by only a small amount. 

\textbf{Our results.} In this work, we consider efficient assembly of shapes in the 2HAM model under the restriction that terminal assemblies are significantly larger than all non-terminal producible assemblies.
We call a system \emph{factor-$c$ size-separable} if the ratio between the smallest terminal assembly and largest non-terminal producible assembly is at least~$c$.
Thus, high-factor size-separable systems lack large but non-terminal assemblies, allowing robust filtering of terminal from non-terminal assemblies in these systems.

Our main result is an algorithm for converting 2HAM systems of a special class into size-separable 2HAM systems.
A 2HAM system $\mathcal{S} = (T, f, \tau)$ consists of a set of \emph{tiles} $T$ that attach by forming bonds according to their \emph{glues} and a \emph{glue-strength function} $f$, and two assemblies can attach if the total strength of the bonds formed meets or exceeds the temperature $\tau$ of the system. 
If a system is temperature-1 ($\tau = 1$), then any two assemblies can attach if they have a single matching glue.
An assembly is said to be \emph{mismatch-free} if no two coincident tile sides in the assembly or any assembly in the system have different glues.
We prove the following:

\vspace*{10pt}

\textbf{Theorem~\ref{thm:mainresult}.}
\emph{Let $\mathcal{S} = (T, f, 1)$ be a 2HAM system with a mismatch-free unique terminal assembly $A$.
Then there exists a factor-2 size-separable 2HAM system $\mathcal{S}' = (T', f', 2)$ with a unique mismatch-free finite terminal assembly $A'$ such that $|\mathcal{S}'| \leq 8|\mathcal{S}|$ and $A'$ has the shape of $A$ scaled by a factor of~2.}

\vspace*{10pt}

Along the way, we prove two results of independent interest on temperature-1 mismatch-free systems.
The \emph{bond graph} of an assembly $A$, denoted $G(A)$, is the dual graph of $A$ formed by a node for each tile, and an edge between two tiles if they form a bond.
We show that any system with a unique mismatch-free finite terminal assembly whose bond graph is not a tree can be made so without increasing the number of tile types in the system:

\vspace*{10pt}

\textbf{Lemma~\ref{lem:treeification} (Tree-ification Lemma).}
\emph{Let $\mathcal{S} = (T, f, 1)$ be a 2HAM system with unique mismatch-free finite terminal assembly $A$.
Then there exists a 2HAM system $\mathcal{S}' = (T', f', 1)$ with unique mismatch-free finite terminal assembly $A'$ and $|\mathcal{S}'| \leq |\mathcal{S}|$, where $A'$ has the shape of $A$ and $G(A')$ is a tree.}

\vspace*{10pt}

The proof of the Tree-ification Lemma yields a simple algorithm for obtaining $\mathcal{S}'$: while a cycle in $G(A)$ remains, remove a glue on this cycle from the tile type containing it.
The challenge is in proving such a process does not disconnect $G(A)$, regardless of the glue and cycle chosen.

We also prove that the tile types used only once in a unique terminal assembly, called \emph{1-occurrence tiles}, form a connected subgraph of $G(A)$.
That is, these tiles taken alone form a valid assembly.

\vspace*{10pt}

\textbf{Lemma~\ref{lem:1occurrence-tiles-1stable}.}
\emph{Let $\mathcal{S} = (T, f, 1)$ be a 2HAM system with unique mismatch-free finite terminal assembly $A$.
Then the 1-occurrence tiles in $A$ form a 1-stable subassembly of $A$.}

\vspace*{10pt}

For some questions about temperature-1 systems, results have been far easier to obtain for mismatch-free systems than for general systems allowing mismatches.
For instance, a lower bound of $2n-1$ for the assembly of a $n \times n$ square by any temperature~1 aTAM system was conjectured by Rothemund and Winfree~\cite{Rothemund-2000a}, and proved for mismatch-free systems by Ma\v{n}uch, Stacho, and Stoll~\cite{Manuch-2010a}.
Meunier~\cite{Meunier-2014b} was able to show the same lower bound for systems permitted to have mismatches under the assumption that the seed tile starts in the lower left of the assembly, and removing this restriction remains open.
In a similar vein, Reif and Song~\cite{Reif-2014a} have shown that temperature-1 mismatch-free aTAM systems are not computationally universal, while the same problem for systems with mismatches permitted is a notoriously difficult problem that remains open, despite significant efforts~\cite{Lathrop-2008a,Doty-2009a,Summers-2010a,Meunier-2014a}.

In spite of such results, constructing high-factor size-separable versions of temperature-1 mismatch-free systems remains challenging.
One difficulty lies in the partitioning the assembly into two equal-sized halves that will come together for the final assembly step.
Note that for many assemblies, such a cutting is impossible (e.g.\ the right assembly in Figure~\ref{fig:simple-examples}).
Even if such a cutting is possible, removing the bonds connecting the two halves by modifying the tiles along the boundary may require a large increase in the number of tile types of the system. 

Another challenge lies in coping with cycles in the bond graph.
Factor-2 size-separability requires that the last assembly step consists of two completely assembled halves of the unique terminal assembly attaching.
Cycles in the bond graph (e.g. the left assembly in Figure~\ref{fig:simple-examples}) prevent communication between the tiles inside and outside of the cycles, risking the possibility that the portion of the assembly inside a cycle still has missing tiles as the exterior takes part in the supposed final assembly step.

\begin{figure}[ht]
\centering
\includegraphics[scale=1.0]{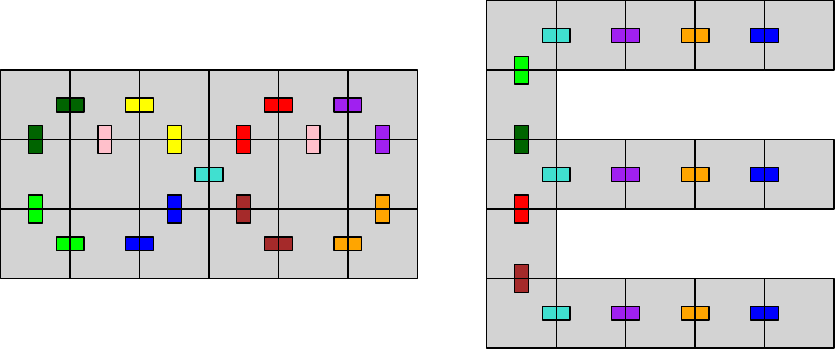}
\caption{Unique mismatch-free terminal assemblies of two different temperature-1 2HAM systems.
Constructing high-factor size-separable versions of these systems is challenging due to the existence of cycles (left) and lack of equal-sized halves (right).}
\label{fig:simple-examples}
\end{figure}

Loosely speaking, our approach is to first construct a version of $A$ where the bond graph is a tree and a vertex cut of $G(A)$ consisting of a path of 1-occurrence tiles exists.
This modified version of $A$ is then scaled in size and temperature by a factor of~2, using special $2 \times 2$ macrotiles that only assemble along the boundary of the scaled assembly via mixed-strength bonds.
Finally, the 1-occurrence tiles forming a vertex cut are given weakened glues such that only completely formed subassemblies on both sides of the cut can attach across the weak-glue cut.

\section{Definitions}
\label{sec:definitions}

Here we give a complete set of formal definitions of tile self-assembly used throughout the paper.
All of the definitions used are equivalent to those found in prior work on the two-handed tile assembly model, e.g.~\cite{Abel-2010a,Cannon-2013a,Chen-2012a,Padilla-2013a}.

\textbf{Assembly systems.} In this work we study the \emph{two-handed tile assembly model (2HAM)}, and instances of the model called \emph{systems}.
A 2HAM system $\mathcal{S} = (T, f, \tau)$ is specified by three parts: 
a \emph{tile set} $T$, a \emph{glue-strength function} $f$, and a \emph{temperature} $\tau \in \mathbb{N}$.

The tile set $T$ is a set of unit square \emph{tiles}.
Each tile $t \in T$ is defined by 4-tuple $t = (g_n, g_e, g_s, g_w)$ consisting of four \emph{glues} from a set $\Sigma$ of \emph{glue types}, i.e. $g_n, g_e, g_s, g_w \in \Sigma$.
The four glues $g_n$, $g_e$, $g_s$, $g_w$ specify the glue types in $\Sigma$ found on the north (N), east (E), south (S), and west (W) sides of $t$, respectively.
Each glue also defines a \emph{glue-side}, e.g. $(g_n, N)$.
Define $g_D(t)$ to be the glue on the side $D$ of $t$, e.g. $g_{\dirN}(t) = g_n$. 

The glue function $f : \Sigma^2 \rightarrow \mathbb{N}$ determines the strength of the \emph{bond} formed by two coincident glue-sides.
For any two glues $g, g' \in \Sigma$, $f(g, g') = f(g', g)$.
A unique \emph{null glue} $\varnothing \in \Sigma$ has the property that $f(\varnothing, g) = 0$ for all $g \in \Sigma$.
In this work we only consider glue functions such that for all $g, g' \in \Sigma$, $f(g, g') = 0$ and if $g \neq \varnothing, f(g, g) > 0$.
For convenience, we sometimes refer to a glue-side with the null glue as a side \emph{without a glue}. 

\textbf{Configurations and assemblies.} A \emph{configuration} is a partial function $C : \mathbb{Z}^2 \rightarrow T$ mapping locations on the integer lattice to tiles.
Define $L_D(x, y)$ to be the location in $\mathbb{Z}^2$ one unit in direction $D$ from $(x, y)$, e.g. $L_{\dirN}(0, 0) = (0, 1)$.
For any pair of locations $(x, y), L_D(x, y) \in C$, the \emph{bond strength} between the these tiles is $f(g_D(C(x, y)), g_{D^{-1}}(C(L_D(x, y))))$.
If $g_D(C(x, y)) \neq g_{D^{-1}}(C(L_D(x, y)))$, then the pair of tiles is said to form a \emph{mismatch}, and a configuration with no mismatches is \emph{mismatch-free}.
If $g_D(C(x, y)) = g_{D^{-1}}(C(L_D(x, y)))$, then the common glue and pair of directions define a \emph{glue-side pair} $(g_D(C(x, y), \{D, D^{-1}\})$.

The \emph{bond graph of $C$}, denoted $G(C)$, is defined as the graph with vertices ${\rm dom}(C)$ and edges $\{ ((x, y), L_D(x, y)) : f(g_D(C(x, y)), g_{D^{-1}}(C(L_D(x, y)))) > 0 \}$.
That is, the graph induced by the neighboring tiles of $C$ forming positive-strength bonds.

A configuration $C$ is a \emph{$\tau$-stable assembly} or an \emph{assembly at temperature $\tau$} if ${\rm dom}(C)$ is connected on the lattice and, for any partition of ${\rm dom}(C)$ into two subconfigurations $C_1$ and $C_2$, the sum of the bond strengths between tiles at pairs of locations $p_1 \in {\rm dom}(C_1)$, $p_2 \in {\rm dom}(C_2)$ is at least $\tau$, the temperature of the system.
Any pair of assemblies $A_1$, $A_2$ are equivalent if they are identical up to a translation by $\langle x, y \rangle$ with $x, y \in \mathbb{Z}$.
The \emph{size} of an assembly $A$ is $|{\rm dom}(A)|$, and $t \in T$ is a \emph{$k$-occurrence tile} in $A$ if $|\{ (x, y) \in {\rm dom}(A) : A(x, y) = t \}| = k$.
The \emph{shape} of an assembly is the polyomino induced by ${\rm dom}(A)$, and a shape is \emph{scaled by a factor $k$} by replacing each cell of the polyomino with a $k \times k$ block of cells. 

Two $\tau$-stable assemblies $A_1$, $A_2$ are said to \emph{assemble} into a \emph{superassembly} $A_3$ if $A_2$ is equivalent to an assembly $A_2'$ such that ${\rm dom}(A_1) \cap {\rm dom}(A_2') = \emptyset$ and $A_3$ defined by the union of the partial functions $A_1$ and $A_2'$ is a $\tau$-stable assembly.
Similarly, an assembly $A_1$ is a \emph{subassembly} of $A_2$, denoted $A_1 \subseteq A_2$, if $A_2$ is equivalent to an assembly $A_2'$ such that ${\rm dom}(A_1) \subseteq {\rm dom}(A_2')$.

\textbf{Producible and terminal assemblies.} An assembly $A$ is a \emph{producible assembly} of a 2HAM system $\mathcal{S}$ if $A$ can be assembled from two other producible assemblies or $A$ is a single tile in $T$.
A producible assembly $A$ is a \emph{terminal assembly} of $\mathcal{S}$ if $A$ is producible and $A$ does not assemble with any other producible assembly of $\mathcal{S}$.

We also consider \emph{seeded} versions of some 2HAM systems, where an assembly is producible if it can be assembled from another producible assembly and a single tile of $T$.
Note that for any temperature-1 2HAM system $\mathcal{S}$, the seeded version of $\mathcal{S}$ has the same set of terminal assemblies as $\mathcal{S}$. 

If $\mathcal{S}$ has a single terminal assembly $A$, we call $A$ the \emph{unique terminal assembly (UTA)} of $\mathcal{S}$.
In the case that $|A|$ is finite and mismatch-free, we further call $A$ the \emph{unique mismatch-free finite terminal assembly (UMFTA)} of $\mathcal{S}$.

\textbf{Size-separability.} A 2HAM system $\mathcal{S} = (T, f, \tau)$ is a \emph{factor-$c$ size-separable} if for any pair of producible assemblies $A$, $B$ of $\mathcal{S}$ with $A$ terminal and $B$ not terminal, $|A|/|B| \geq c$. 
Since this ratio is undefined when $\mathcal{S}$ has infinite producible assemblies, such systems have undefined size-separability.
Every system with defined size-separability has factor-$c$ size-separability for some $1 \leq c \leq 2$. 

\section{Tree-ification}
\label{sec:tree-ification}

First, we prove that any $\tau=1$ system producing a unique terminal assembly can be converted into a system with another unique terminal assembly with the same shape but whose bond graph is a tree.
This is formalized in the Tree-ification Lemma (Lemma~\ref{lem:treeification}) at the end of this section.

\begin{lemma}
\label{lem:1-stable-is-producible}
Let $\mathcal{S} = (T, f, 1)$ be a 2HAM system.
Every 1-stable assembly consisting of tiles in $T$ is a producible assembly of $\mathcal{S}$.
\end{lemma}

\begin{proof}
Let $B$ be a 1-stable assembly consisting of tiles in $T$.
Perform a breadth-first traversal of $G(B)$, starting at an arbitrary tile, to obtain an ordering on the tiles of $A$.
Now consider the set of assemblies $\{ A_n \}$ consisting of the first $n$ tiles reached in the breadth-first traversal for $1 \leq n \leq |B|$.
The assembly $A_1$ is the root tile of the breadth-first traversal and so is trivially producible.
Assume that the first $n$ assemblies are producible.
Assembly $A_{n}$ is a superassembly of $A_{n-1}$ and the the tile reached in step $n$ of the breadth-first search.
So by induction all assemblies $\{ A_n \}$, including $A_{|B|} = B$, are producible by $\mathcal{S}$.
\end{proof}

\begin{lemma}
\label{lem:repeated-glue-side-pair-case-1}
Let $\mathcal{S} = (T, f, 1)$ be a 2HAM system with UTA $A$. 
Let a glue-side pair appear twice on a simple cycle of $G(A)$ between tiles $t_1$ and $t_2$, and $t_3$ and $t_4$.
Then $|\{t_1, t_2, t_3, t_4\}| \neq 4$. 
\end{lemma}

\begin{proof}
We prove the result by contradiction.
Without loss of generality, assume the glue-side pair is $(1, \{\dirE, \dirW\})$, so $t_1$ and $t_3$ are west of $t_2$ and $t_4$, respectively.
Consider the seeded version of $\mathcal{S}$ additionally restricted in two ways: any producible assembly with $t_1$ and no tile east of $t_1$ must immediately attach $t_4$ east of $t_1$, and any producible assembly with $t_2$ and no tile west of $t_2$ must immediately attach $t_3$ west of $t_2$.

This seeded version of $\mathcal{S}$ has the same set of terminal assemblies as $\mathcal{S}$, and so has a unique terminal assembly $A$.
However, the assembly contains no occurrences of $t_1$ west of $t_2$ and so cannot be $A$, a contradiction.
\end{proof}

\begin{lemma}
\label{lem:repeated-glue-side-pair-case-2}
Let $\mathcal{S} = (T, f, 1)$ be a 2HAM system with UMFTA $A$.
Let a glue-side pair appear twice on a simple cycle of $G(A)$ between tiles $t_1$ and $t_2$, and $t_3$ and $t_4$.
Then $|\{t_1, t_2, t_3, t_4\}| \neq 2$. 
\end{lemma}

\begin{proof}
We prove the result by contradiction.
First, observe that if $|\{t_1, t_2, t_3, t_4\}| = 2$, then $t_1 = t_3$ and $t_2 = t_4$, otherwise $\mathcal{S}$ produces an infinite assembly from just this set of four tiles.
Also, $t_1$ and $t_2$ must appear in the same relative positions in both occurrences of the glue-side pair, otherwise $\mathcal{S}$ produces an infinite assembly.

\textsc{Growing a second cycle.} 
We carry out seeded assembly, starting with the assembly $C_1$ consisting of tiles on the cycle.
Recall that $G(C_1)$ is a cycle with two occurrences of the adjacent tile pair $t_1$ and $t_2$.
Starting at the second occurrence of the pair, we attach a sequence the single tiles to $C_1$ occurring along the cycle $G(C_1)$, starting at the first occurrence of the pair (see Figure~\ref{fig:pumping-pf-1}).
The sequence finishes with reaching the first occurrence of $t_1$ and $t_2$ again, having made a complete tour of tile attachments they occur along the cycle $G(C_1)$.
We call the resulting assembly $C_2$.

\begin{figure}[ht]
\centering
\makebox[\textwidth][c]{\includegraphics[scale=1.0]{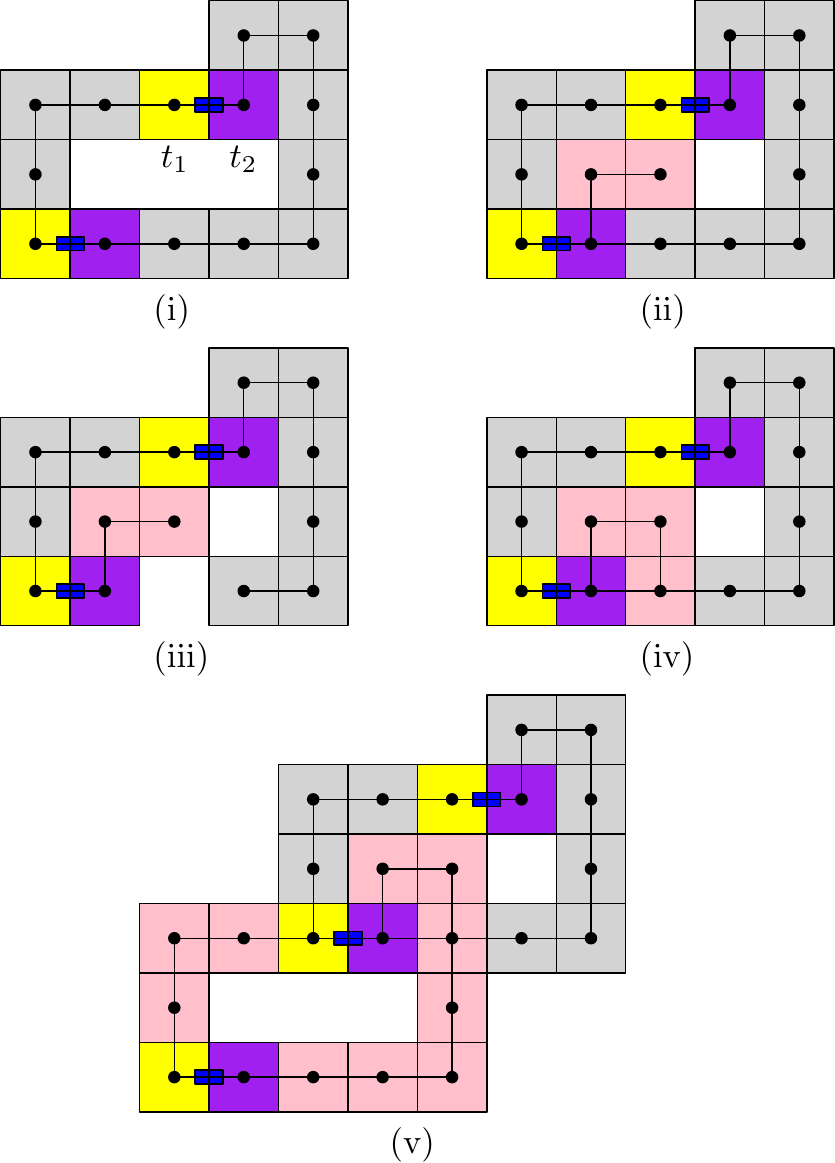}}
\caption{Assembling $C_2$ (step~(v)), starting with $C_1$ (step~(i)).
In steps~(iii) and~(iv), a blocking tile is replaced with the next tile along the cycle.
The resulting bond graph has a spanning subgraph of two cycles sharing two vertices.}
\label{fig:pumping-pf-1}
\end{figure}

Because the cycle is simple and thus non-self-intersecting, no tile attachment is prevented (blocked) by the presence of a tile appearing earlier along the cycle.
However, tile attachments may be blocked by the presence of a tile of $C_1$.
In this situation, we replace the blocking tile (called $t_{\rm old}$) of $C_1$ with the next tile along the cycle ($t_{\rm new}$).

Just before removing $t_{\rm old}$, the assembly is spanned by the tiles of $C_1$ and the path of tiles attached so far.
The cycle and path share two tiles: the second occurrence of $t_1$ and $t_2$ in $C_1$. 
Removing $t_{\rm old}$ yields a 1-stable assembly, and placing $t_{\rm new}$ yields another producible assembly.
By Lemma~\ref{lem:1-stable-is-producible}, both of these assemblies are producible assemblies of $\mathcal{S}$. 
Since $A$ has no mismatches, $t_{\rm new}$ attaches to both neighbors of $t_{\rm old}$ along the simple cycle of $G(C_1)$.

We repeat this replacement process every time a blocking tile is encountered, and attach tiles until the first occurrence of $t_1$ and $t_2$ is reached.
At this point, the path is closed to form a second cycle and we call the assembling resulting from \emph{growing} this second cycle $C_2$.

The sequence of tiles around $C_1$ form a cycle in $G(C_2)$, and the newly attached sequence of tiles form a second cycle in $G(C_2)$.
The cycles also share a common pair of vertices: the two tiles $t_1$ and $t_2$ in the second occurrence of the glue-side pair in $C_1$.

\textsc{Cycle pumping.} 
We now grow additional cycles indefinitely to produce an infinite sequence of producible assemblies $\{ C_n \}$ of $\mathcal{S}$.
In the $n$th repetition, the second occurrence of the glue-side pair on the cycle formed in the $(n-1)$st repetition is used as the starting point for placing the tiles as they appear around the cycle.
Assume by induction that the assembly at the start of the $n$th repetition, $G(C_{n-1})$, has a subgraph consisting of a sequence of $n-1$ cycles, where each cycle in the sequence shares two vertices with adjacent cycles in the sequence, and all vertices belong to some cycle.
The sequence of tile placements to produce $C_n$ then extends this graph with a path.

Replacing a blocking tile with the next tile along the $n$th cycle is always possible, as removing the tile removes at most one vertex from each cycle and disconnecting the graph requires removing at least two vertices from a cycle.
At the end of the sequence of placements, the bond graph $G(C_n)$ of the resulting assembly consists of a sequence of $n$ cycles, each sharing a pair of vertices with adjacent cycles in the sequence. 

The set of locations of tiles in $C_n$ is $\{i \cdot (x_{o_1} - x_{o_2}, y_{o_1} - y_{o_2}) + (x, y) \mid 0 \leq k \leq i, (x, y) \in {\rm dom}(C_1) \}$, where $(x_{o_1}, y_{o_1})$ and $(x_{o_2}, y_{o_2})$ are the locations of $t_1$ in the first and second occurrences of the glue-side pairs, respectively. 
So $|C_n| > |C_{n-1}|$ and the $n$th cycle contains at least one vertex not found in any previous cycle. 
So $\{C_n\}$ contains arbitrarily large assemblies producible by $\mathcal{S}$, and $A$ cannot be the UMFTA of $\mathcal{S}$, a contradiction. 
\end{proof}

\begin{lemma}
\label{lem:repeated-glue-side-pair-case-3}
Let $\mathcal{S} = (T, f, 1)$ be a 2HAM system with UTA $A$.
Let a glue-side pair appear twice on a simple cycle of $G(A)$ between tiles $t_1$ and $t_2$, and $t_3$ and $t_4$.
Then $|\{t_1, t_2, t_3, t_4\}| \neq 3$.
\end{lemma}

\begin{proof}
Assume without loss of generality that the glue-side pair is $(1, \{\dirE, \dirW\})$ 
Let $t_1$ and $t_3$ be west of $t_2$ and $t_4$, respectively.
If $t_1 = t_4$ or $t_2 = t_3$, then $\mathcal{S}$ produces an infinite assembly, so either $t_1 = t_3$ or $t_2 = t_4$; assume without loss of generality that $t_2 = t_4$.

Consider the subassembly $C$ of $A$ consisting only of the tiles forming the cycle.
Since $G(C)$ is a cycle, removing $t_1$ from $C$ and replacing it with $t_3$ yields a 1-stable (and thus producible by Lemma~\ref{lem:1-stable-is-producible}) assembly $C'$ of $\mathcal{S}$.
But $C'$ has two occurrences of the same glue-side pair between tiles $t_3$ and $t_2$, contradicting Lemma~\ref{lem:repeated-glue-side-pair-case-2}.
\end{proof}

\begin{lemma}
\label{lem:no-repeated-glue-in-cycle}
Let $\mathcal{S} = (T, f, 1)$ be a 2HAM system with UMFTA $A$.
Then no glue-side pair appears twice on a simple cycle of $G(A)$. 
\end{lemma}

\begin{proof}
Suppose that some glue-side pair appears twice, and let the four tiles of the two occurrences be $t_1$, $t_2$, $t_3$, and $t_4$.
Consider $k = |\{t_1, t_2, t_3, t_4\}| \in \{1, 2, 3, 4\}$.
Clearly $k \neq 1$, as otherwise $\mathcal{S}$ produces an infinite assembly. 
By Lemma~\ref{lem:repeated-glue-side-pair-case-1}, $k \neq 4$.
By Lemma~\ref{lem:repeated-glue-side-pair-case-2}, $k \neq 2$.
Finally, Lemma~\ref{lem:repeated-glue-side-pair-case-3} implies $k \neq 3$.
So no glue-side pair can occur twice on a simple cycle of $G(A)$.
\end{proof}

\begin{lemma}
\label{lem:all-occurrences-on-off-cycle}
Let $\mathcal{S} = (T, f, 1)$ be a 2HAM system with UMFTA $A$.
Let $(g, p)$ be the glue-side pair of an edge $e$ in $G(A)$.
Then if $e$ lies on a simple cycle in $G(A)$, all edges with glue-side pair $(g, p)$ lie on simple cycles of $G(A)$. 
\end{lemma}

\begin{proof}
Consider a seeded version of $\mathcal{S}$ modified in the following way: each time an attaching tile $t$ leaves $g$ exposed on a side in $p$, carry out the sequence of tile attachments as they occur on the cycle in $A$ containing $e$.
Let $B$ be the assembly as it appears just after $t$ is placed.
The first of these attachments is the one of the two tiles forming the vertices of $e$ in $G(A)$.
Continue the attachments until attaching the next tile on the cycle is blocked by an existing tile (see Figure~\ref{fig:all-instances-on-cycles}).

\begin{figure}[ht]
\centering
\includegraphics[scale=1.0]{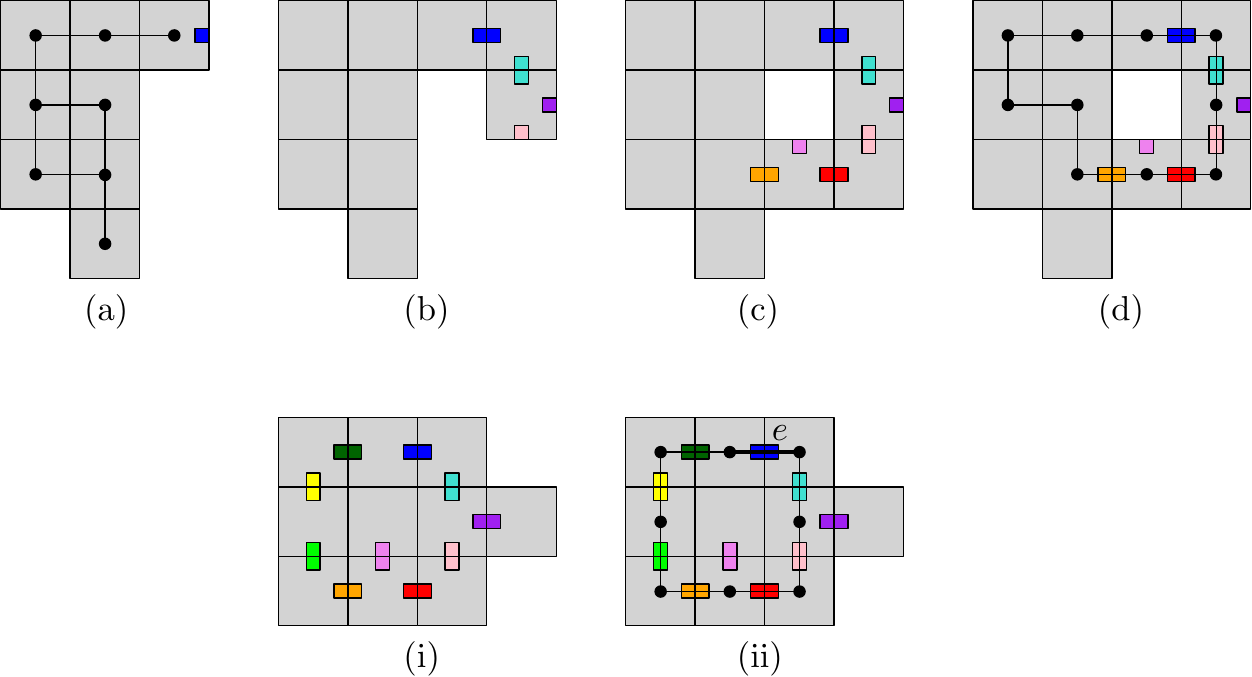}
\caption{An example demonstrating the proof of Lemma~\ref{lem:all-occurrences-on-off-cycle}. 
Part~(a): growing $A$ using a seeded assembly process, the glue-side $(\rm{blue}, \dirE)$ appears.
Part~(b): immediately a sequence of single-tile attachments following those found on a cycle of $A$ (parts~(i) and~(ii)) are placed.
Part~(c): the next tile in the cycle can no longer be placed, and a bond between the assembly of part~(a) and the growing cycle is formed.
Part~(d): the resulting cycle consisting of paths through the sequence of single-tile attachments and the assembly before the sequence of attachments.}
\label{fig:all-instances-on-cycles}
\end{figure}

Each tile attached leaves an exposed glue matching a glue of the next tile in the cycle, so the blocking tile must form a bond with the last attached tile.
Since the sequence of tile attachments follows a simple cycle in $A$, the blocking tile must have been part of $B$.
So the final attachment results in a bond between the sequence of tile placements (part~(c) of Figure~\ref{fig:all-instances-on-cycles}) and $B$ (part~(a) of Figure~\ref{fig:all-instances-on-cycles}).
Thus a cycle is formed by a path starting at $t$ through the sequence of single-tile attachments and a path through $B$ ending at $t$.

Because this modified version of $\mathcal{S}$ only modified the order in which tiles attach, it produces a unique terminal assembly $A$.
So all edges of $A$ with glue-side pair $(g, p)$ lie on simple cycles of $G(A)$. 
\end{proof}

\begin{lemma}[Tree-ification Lemma]
\label{lem:treeification}
Let $\mathcal{S} = (T, f, 1)$ be a 2HAM system with UMFTA $A$.
Then there exists a 2HAM system $\mathcal{S}' = (T', f', 1)$ with UMFTA $A'$ and $|\mathcal{S}'| \leq |\mathcal{S}|$, where $A'$ has the shape of $A$ and $G(A')$ is a tree. 
\end{lemma}

\begin{proof}
The system $\mathcal{S}'$ and assembly $A'$ are created by repeatedly breaking cycles in the bond graph of $A$.
A cycle is broken by removing an edge of the cycle, i.e. an occurrence of a glue-side pair.
The break is implemented in the tile set by replacing all occurrences of glues of the glue-side pair with the null glue.

Such an operation could potentially disconnect $G(A')$, causing $A'$ to no longer be 1-stable (and producible) assembly of $\mathcal{S}'$.
This could happen in two ways: either the glue-side pair occurs twice on a cycle, or the glue-side pair is a cut edge of the bond graph of $A'$ elsewhere.
By Lemma~\ref{lem:no-repeated-glue-in-cycle}, no glue-side pair appears twice in any cycle.
So the first possibility cannot occur.
By Lemma~\ref{lem:all-occurrences-on-off-cycle}, if a glue-side pair appears on a cycle, then it cannot be a single-edge cut anywhere in $A$. 
So removing all occurrences of a glues on the sides of the glue-side pair from $\mathcal{S}'$ does not cause $A'$ to become 1-unstable.
Moreover, removing glues cannot add to the set of producible assemblies of $\mathcal{S}'$, so $A'$ remains the unique terminal assembly of $\mathcal{S}$.

Returning to cycle breaking, each break decreases the number of edges in $G(A')$, so this process must terminate.
At termination, $G(A')$ has no cycles, and so is a tree.
Since we only removed glues from tiles of $\mathcal{S}$ to create $\mathcal{S}'$, the tile set is not larger.
However, the removal of glues may cause multiple tile in $\mathcal{S}$ to be indistinguishable in $\mathcal{S}'$.
So $|\mathcal{S}'| \leq |\mathcal{S}|$.
\end{proof}

\section{1-Occurrence Tile Types}
\label{sec:1-occurrence}

In addition to tree-ification, we also make use of the existence of \emph{1-occurrence tile types}: tile types that appear only once in the terminal assembly of the system.
These special tile types are utilized to modify the unique terminal assembly to contain unique tile types at two locations equally spaced along a traversal around the outside of the assembly.

\begin{lemma}
\label{lem:at-least-two-1-occurrence-tile-types}
Let $\mathcal{S} = (T, f, 1)$ be a 2HAM system with UMFTA $A$ with $G(A)$ a tree and $|A| \geq 2$.
Then $A$ has at least two 1-occurrence tiles.
\end{lemma}

\begin{proof}
We prove the result by induction on $|T|$, and assume that each tile in $T$ has at least once occurrence in $A$.
We use $|T| = 2$ as a base case, since if $|A| \geq 2$, then $|T| \geq 2$. 
Since $|A| = |T|$, both tiles in $A$ are 1-occurrence tiles and the claim holds.

For the inductive step, we examine the leaf tiles of $A$ and either find two 1-occurrence tiles or create a new system $\mathcal{S}' = (T', f', 1)$ with a UMFTA $A'$ with $G(A')$ a tree, $|T'| < |T|$, and the same number of 1-occurrence tiles as in $A$.

Recall that any tree with at least two nodes, including $G(A)$, has at least two leaves.
If two leaf tiles are 1-occurrence tiles, then the claim holds.
Otherwise at least one leaf tile $t$ is a $k$-occurrence tile with $k \geq 2$.

Since $A$ is mismatch-free, $t$ must have only one non-null glue $g$ and all occurrences of $t$ are leaves of $G(A)$.
Without loss of generality, assume $g = g_{\dirN}(t)$, i.e. $g$ is on the north side of $t$.
The glue $g$ is not the north glue of any other tile in $T$, as otherwise $t$ could be replaced with this other tile to yield a producible assembly with size $|A|$ not equal to $A$, a contradiction.
Moreover, since all non-north glues of $t$ are null, the tiles to the east, south, and west of each occurrence of $t$ must have the null glue on their west, north, and east sides, respectively.
So removing all occurrences of $t$ from $A$ yields a producible assembly $A'$ for the system $\mathcal{S}' = (T - \{t\}, f, 1)$.

Since all occurrences of $t$ in any producible assembly are leaves, removing all occurrences of $t$ from an assembly yields another 1-stable assembly.
So an assembly $A_{\mathcal{S}}$ is a producible assembly of $\mathcal{S}$ if and only if $A_{\mathcal{S}}$ with all occurrences of $t$ removed is a producible assembly of $\mathcal{S}'$.
So $A'$ must be the UTA of $\mathcal{S}'$. 
Also, $A'$ has the same number of 1-occurrence tiles as $A$.
So by induction, the mismatch-free unique terminal assembly $A'$ of $\mathcal{S}'$ has at least two 1-occurrence tiles and so does $A$.
\end{proof}

\begin{lemma}
\label{lem:1occurrence-tiles-1stable}
Let $\mathcal{S} = (T, f, 1)$ be a 2HAM system with UMFTA $A$.
Then the 1-occurrence tiles in $A$ form a 1-stable subassembly of $A$.
\end{lemma}

\begin{proof}
Suppose by contradiction that the 1-occurrence tiles in $A$ do not form a 1-stable subassembly.
Then let $t_1$ and $t_2$ be a pair of 1-occurrence tiles such that $t_1$ and $t_2$ are not connected with a path of 1-occurrence tiles in $G(A)$ and the shortest path between $t_1$ and $t_2$ is as short as possible.
In other words, $t_1$ and $t_2$ are the closest pair of 1-occurrence tiles that are not part of a 1-stable subassembly of $A$ consisting only of 1-occurrence tiles. 

Let $t_3$ and $t_4$ be the second and second-to-last tiles encountered along the shortest path in $G(A)$ from $t_1$ to $t_2$.
By definition, $t_3$ and $t_4$ are not 1-occurrence tiles.
Without loss of generality, assume the glue-side pair between $t_1$ and $t_3$ is $(1, \{\dirW, \dirE\})$, with $t_1$ west of $t_3$, and the glue-side pair between $t_4$ and $t_2$ is $(2, \{\dirW, \dirE\})$, with $t_4$ west of $t_2$.
Let $B$ be the 1-stable subassembly of $A$ consisting of the aforementioned path between $t_1$ and $t_2$ (including an occurrence $t_4$) and the shortest path in $G(A)$ between the two occurrences of $t_4$.
The remainder of the proof is a case analysis of $B$, proving that the existence $B$ implies another 1-stable assembly with two occurrences of $t_1$ or $t_2$.
The three cases are seen in Figure~\ref{fig:single-occurrence-tiles-connected}.

\begin{figure}[ht]
\centering
\makebox[\textwidth][c]{\includegraphics[scale=1.0]{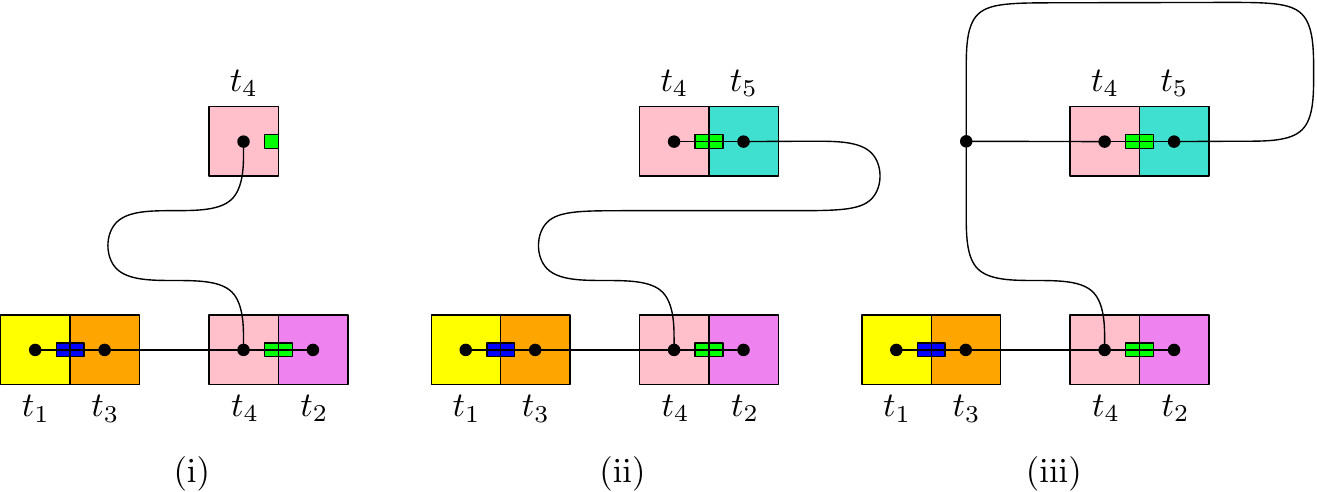}}
\caption{The three cases for the closest pair of 1-occurrence tiles ($t_1$ and $t_2$) in $A$ not in the same connected component of 1-occurrence tiles in $G(A)$.
The shortest path between two occurrences of $t_4$ is used to create a 1-stable assembly containing two occurrences of $t_1$ or $t_2$.}
\label{fig:single-occurrence-tiles-connected}
\end{figure}

\textbf{Case~(i).} Suppose that the path between the two occurrences of $t_4$ does not contain a tile in the location east of the second occurrence of $t_4$.
Attaching $t_2$ at this location yields a 1-stable assembly containing two occurrences of $t_2$.

\textbf{Case~(ii).} Now suppose that the path between the two occurrences of $t_4$ contains a tile $t_5$ in the location east of the second occurrence of $t_4$.
Begin attaching tiles as they appear along the path between $t_2$ and $t_1$, starting at $t_4$.
If all tiles along the path can be attached, then the result is a 1-stable assembly containing two occurrences of $t_1$.

\textbf{Case~(iii).} Assume the situation is identical to case~(ii), but some tile placement along the path from $t_2$ to $t_1$ is blocked by an existing tile in the assembly.
Since the assembly is producible, it must be a subassembly of $A$ and thus mismatch-free.
So the previous tile along the path and the blocking tile must share a glue-side pair.
Moreover, this glue must have positive strength since the glue was sufficient for the next tile along the path from $t_2$ to $t_1$ to attach.
So a cycle containing adjacent tiles $t_4$ (the second occurrence) and $t_5$ is formed by a portion of the path between the two occurrences of $t_4$ and the partial path from $t_2$ to $t_1$.
So removing $t_5$ yields a 1-stable assembly containing both $t_2$ and an occurrence of $t_4$ with no tile to the east, and replacing $t_5$ with $t_2$ yields a 1-stable assembly containing two occurrences of $t_2$.

So in all three cases, a 1-stable and thus producible (by Lemma~\ref{lem:1-stable-is-producible}) assembly is created with two occurrences of a tile that appears only once in $A$.
So $A$ cannot be the unique terminal assembly of $\mathcal{S}$, a contradiction.
\end{proof}

\begin{lemma}
\label{lem:1-occurrence-glues-unique}
Let $\mathcal{S} = (T, f, 1)$ be a 2HAM system with UTA $A$. 
For any glue-side pair $(g, p)$ occurring between a pair of 1-occurrence tiles in $A$, $(g, p)$ occurs only once in $A$. 
\end{lemma}

\begin{proof}
Without loss of generality, let $(g, p) = (g, \{\dirE, \dirW\})$ and $t_1$ and $t_2$ be the pair of 1-occurrence tiles with the glue-side pair occurrence with $t_1$ west of $t_2$.
Consider the seeded version of $\mathcal{S}$ additionally restricted in two ways: any time a tile attaches to a producible assembly $P$ with exposed east glue $g$, $t_2$ must immediately attach to $P$ via this glue; similarly, any time a tile attaches to $P$ with exposed west glue $g$, $t_1$ must immediately attach to $P$.
So every occurrence of the glue-side pair $(g, p)$ contains $t_1$ or $t_2$.
Then since $A$ is a producible assembly of $\mathcal{S}$ and $t_1$ and $t_2$ are 1-occurrence tiles, $(g, p)$ occurs only once in $A$. 
\end{proof}

\begin{lemma}
\label{lem:two-occurrences-same-glue-side} 
Let $\mathcal{S} = (T, f, 1)$ be a 2HAM system with UMFTA $A$ with $G(A)$ a tree. 
For any tile $t \in T$, the simple path in $G(A)$ between any two occurrences of $t$ uses the same glue-side of $t$ on both occurrences.
\end{lemma}

\begin{proof}
Suppose, by contradiction, that there exists a simple path in $G(A)$ between two occurrences of a tile $t$ such that the path uses, without loss of generality, the north glue-side of the first occurrence and the east glue-side of the second occurrence.
Start with the 1-stable assembly formed by the path, including both occurrences of $t$.
Carry out a sequence of tile attachments from the second occurrence of $t$ as they appear along the path from the first to the second occurrence of $t_1$, including the final attachment of $t$.
Now repeat this sequence of attachments indefinitely, starting from the occurrence of $t$ just attached, to create an assembly $B$.

The assembly $B$ is 1-stable and thus producible (by Lemma~\ref{lem:1-stable-is-producible}), and so $B$ is a subassembly of $A$.
Since $A$ is mismatch-free and $G(A)$ is a tree, $B$ is mismatch-free and $G(B)$ is a tree.
So the repeated sequence of attachments cannot be blocked by an existing tile in the assembly and $B$ is an infinite assembly.
So $A$ is not the unique terminal assembly of $\mathcal{S}$, a contradiction.
\end{proof}

\begin{lemma}
\label{lem:1-occurrence-path-wherever}
Let $\mathcal{S} = (T, f, 1)$ be a 2HAM system with UMFTA $A$ with $G(A)$ a tree.
Let edges $e, e' \in G(A)$, with $e'$ between a pair of 1-occurrence tiles.
Then there exists a second 2HAM system $\mathcal{S}' = (T', f', 1)$ with $|T'| \leq 2|T|$ and UMFTA $A'$ with $G(A') = G(A)$ and the unique path from $e'$ to $e$ in $G(A')$ consisting entirely of 1-occurrence tiles in $A'$. 
\end{lemma}

\begin{proof}
The construction of $\mathcal{S}'$ is simple.
Note that $e'$ has a unique glue-side pair by Lemma~\ref{lem:1-occurrence-glues-unique}.
Find the shortest path in $G(A)$ from $e'$ to $e$.
Replace all glues along this path, including $e$, with unique glues, letting $T'$ and $f'$ be the updated tile set and glue function. 

We claim that the resulting mismatch-free assembly $A'$ with $G(A')$ a tree is the unique terminal assembly of $\mathcal{S}'$.
Certainly $A'$ is a terminal assembly of $\mathcal{S}'$, as otherwise $A$ was not a terminal assembly of $\mathcal{S}$.

For any terminal assembly of $\mathcal{S}'$, all occurrences of newly created tiles in $T' - T$ can be swapped with the tiles in $T$ they replaced to yield $A$.
Then since $A'$ is mismatch-free with $G(A')$ a tree and the glues along the path from $e'$ to $e$ in $G(A')$ are unique, any appearance of the tiles along this path in a terminal assembly of $\mathcal{S}'$ must lie along a complete path from $e'$ to $e$.
Otherwise replacing the newly created tiles yields a terminal assembly of $\mathcal{S}$ with a mismatch or a cycle.
In summary, $A'$ must be the unique terminal assembly of $\mathcal{S}'$, as every terminal assembly must have a complete path from $e'$ to $e$ for every occurrence of a tile on the path, swapping all newly created tiles with the tiles they replaced must yield $A$, and the tiles incident to $e'$ are 1-occurrence tiles in $A$.

We also claim that $|T'| \leq 2|T|$ by now proving that the path of tiles between $e'$ and $e$ in $A$ are all distinct.
Suppose that two tiles along the path have the same type $t$.
By Lemma~\ref{lem:two-occurrences-same-glue-side}, the path between two occurrences of $t$ must enter both occurrences using the same glue-side.
So consider the seeded assembly process starting at the tiles incident to $e'$, growing along the path to the first occurrence of $t$, then to the second occurrence, and then growing from the second occurrence using the sequence of tile placements encountered when traveling from the first occurrence \emph{backwards} to $e'$.
Since $A$ is the mismatch-free unique terminal assembly of $\mathcal{S}$ and $G(A)$ is a tree, this sequence of tile placements must cannot be blocked or form a cycle.
So the resulting assembly, a subassembly of $A$, has two occurrences of the 1-occurrence tiles incident to $e'$, a contradiction.
Lemma~\ref{lem:1-occurrence-glues-unique} implies the final result that the edges of $G(A')$ corresponding to the cut are unique glue-side pairs. 
\end{proof}

\section{A Size-Separable Macrotile Construction}
\label{sec:macrotile-construction}

A simple barrier to general high-factor size-separability is the fact that any system with a tree-shaped unique terminal assembly $A$ cannot be factor-$c$ size-separable for any $c > 1 + 1/|A|$. 
A more subtle challenge is how to partition assemblies into equal-sized 1-stable halves that will come together in the final assembly step.
For instance, every 1-stable subassembly of the right assembly of Figure~\ref{fig:simple-examples} has size at most one-third the size of the total assembly.

We resolve both of these issues by creating a temperature-2 2HAM system with a unique terminal assembly whose shape is the shape of $A$ scaled by a factor of~2, and whose bond graph has an edge cut of two temperature-1 bonds that partitions $G(A)$ into two subgraphs of equal size. 

\begin{lemma}
\label{lem:subassembly-of-producible-is-not-terminal}
Let $\mathcal{S} = (T, f, \tau)$ be a 2HAM system with $P$ and $P'$ producible assemblies of $\mathcal{S}$ with $P$ a proper subassembly of $P'$.
Then $P$ is not a terminal assembly.
\end{lemma}

\begin{proof}
Doty~\cite{Doty-2013a} proves that the greedy process of starting with the tiles of $P'$ and repeatedly assembling pairs of $\tau$-stable assemblies into larger $\tau$-stable assemblies always suffices to yield $P'$, regardless of the order in which the assemblies are merged. 
Since $P$ is producible, there must also be such a sequence for $P$, and since $P$ is a subassembly of $P'$, any such sequence utilizes a subset of the tiles in $P'$.
Then $P'$ can be assembled by first assembling $P$, then continuing the process to form $P'$, so $P$ is not terminal.
\end{proof}

\begin{lemma}
\label{lem:scaling}
Let $\mathcal{S} = (T, f, 1)$ be a 2HAM system with UMFTA $A$ with $G(A)$ a tree.
Then there exists a 2HAM system $\mathcal{S}' = (T', f', 2)$ with UMFTA $A'$ and $|\mathcal{S}'| \leq 4|\mathcal{S}|$ such that $A'$ has the shape of $A$ scaled by a factor of~2.
\end{lemma}

\begin{proof}
We start by describing common properties of all occurrences of each tile type $t \in T$.
Since $G(A)$ is a tree, Lemmas~\ref{lem:at-least-two-1-occurrence-tile-types} and~\ref{lem:1occurrence-tiles-1stable} imply that there exists an edge $e'$ in $G(A)$ between two 1-occurrence tiles and Lemma~\ref{lem:two-occurrences-same-glue-side} implies that any path between two occurrences of $t$ use the same glue-side pair.
So any breadth-first search $G(A)$ starting at a 1-occurrence tile incident to $e'$ visits all occurrences of $t$ exactly once, and all via incoming edges from the same side of $t$.
Then since $A$ is mismatch-free, if a direction is applied to each edge of $G(A)$ according to the direction of traversal during the breadth-first search, all occurrences of $t$ have the same set of incoming and outgoing edges.
So all occurrences of $t$ have their corners visited in the same order during a traversal of the boundary of $A$.

We use these conditions to construct unique macrotile versions of each tile type according to their incoming and outgoing edges induced by the breadth-first search starting at $e'$. 
All possible macrotile constructions (up to symmetry) are shown in Figure~\ref{fig:macrotile}.
For each glue-side pair in the original system, we use two glue-side pairs in the scaled system, one with strength-2 and the other with strength-1.
The glue-side pair visited first in the counterclockwise traversal of the boundary starting at $e'$ has strength~2, while the other pair has strength~1.
 
\begin{figure}[ht]
\centering
\includegraphics[scale=1.0]{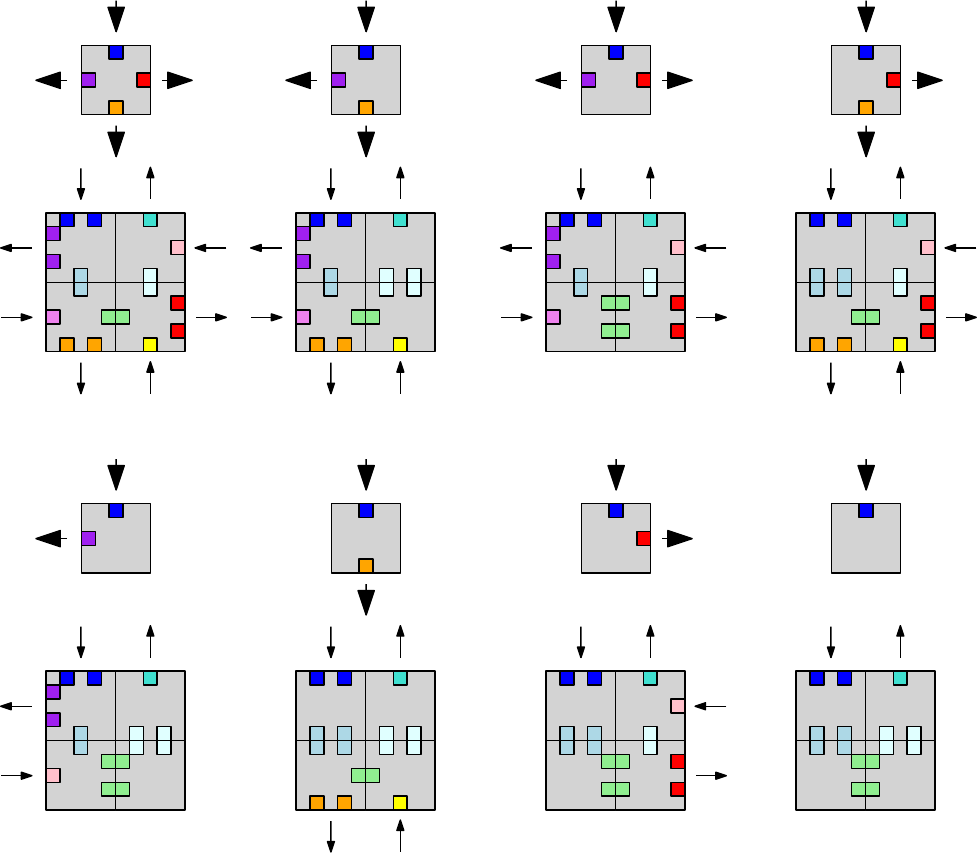}
\caption{The individual tiles enumerate (up to symmetry) all combinations of incoming and outgoing edges (large arrows) induced by a breath-first search of $G(A)$.
The corresponding $2 \times 2$ macrotiles are used in the proof of Theorem~\ref{thm:mainresult} to construct a temperature-2 system that carries out the assembly of $A$ at scale~2 in the order that the tiles appear along the boundary (small arrows).
All internal macrotile glues are unique to the tile type, while all external macrotile glues correspond to the glues found on the surface of the inducing tile.}
\label{fig:macrotile}
\end{figure}

There are also three glues internal to each macrotile attaching each pair of adjacent tiles forming a macroedge whose corresponding edge of the tile either has an outgoing edge induced by the breadth-first search, or has no edge.
These glues are unique to the macrotile type.
The strengths of each of these glues is determined by whether the closest macroside has glues.
If not, then the glue is strength-2, otherwise the glue is strength-1.

If the closest macroside does not contain a glue, then the strength-2 internal glue is necessary to allow assembly to continue along the boundary of the assembly in the counterclockwise direction (e.g.\ from the northwest to the southwest tile).
If the closest macroside does contain a glue, then the strength-1 internal glue prevents the placement of the next tile in the macrotile (e.g.\ the southwest tile after the northwest tile) until a second tile from an adjacent macrotile (e.g.\ the southeast tile of the macrotile to the west) has been placed.
As a result, no pair of tiles in a macrotile can be present in a common assembly unless all tiles between them along a counterclockwise traversal of the boundary of the macrotile assembly are also present.

\textsc{Scaled assembly of $A$.} We claim that this scaled version of the system has a unique terminal assembly $A'$ obtained by replacing each tile in the original unique terminal assembly with the corresponding $2 \times 2$ macrotile.
First we prove that any subassembly of $A'$ corresponding to a subtree of $G(A)$ is producible.
A subtree of size~1 corresponds to a leaf node, the lower-rightmost case in Figure~\ref{fig:macrotile}, and is clearly producible.
For larger subtrees, the assembly can be formed by combining the~4 tiles of the root macrotile to the (up to) three subtrees assemblies.
Grow the assembly in counterclockwise order around the boundary, attaching either a subtree assembly (if the macroside has a glue) or the next tile of the root macrotile.
In both cases, placing the second root tile along the macroedge is possible, as either the internal glue shared with the previous root tile is strength-2 or a second glue is provided by the subtree assembly.
Then by induction, the assembly $A'$ corresponding to the subtree rooted at the root of the breadth-first search is producible.

By construction, $A'$ is terminal because it corresponds to a mismatch-free terminal assembly in the original system that necessarily has no exposed glues.
So $A'$ is a terminal assembly of the scaled system.
Next, we prove that $A'$ is the unique terminal assembly of the system.

\textsc{Terminal assembly uniqueness.} We start by proving that every producible assembly can positioned on a $2 \times 2$ macrotile \emph{grid}, where every tile in the southwest corner is a southwest tile of some macrotile, every tile in the northwest corner is the northwest tile of some macrotile, etc.
Start by noticing that each glue type appears coincident to only one of 12 edges of the grid: the~4 internal edges of each macrotile, and the~8 external edges.
Suppose there is some smallest producible assembly that does not lie on a grid.
Then this assembly must be formed by the attachment of two smaller assemblies that do lie on grids, and whose glues utilized in the attachment are coincident to only one of~12 edges of the grid.
So if these assemblies are translated to have coincident matching glue sides, then their grids must also be aligned and the assembly resulting from their attachment also lies on the grid, a contradiction.

Let $A_p'$ be a producible assembly of the macrotile system that is not $A'$.
Construct an assembly $A_p$ of the original input system $\mathcal{S}$ in the following way: replace each macrotile region with a single tile corresponding to one of the tiles in the macrotile region.
If such a replacement is unambiguous, meaning that all tiles in each macrotile region belong to a common macrotile, then the resulting assembly is a 1-stable (and thus producible) assembly of $\mathcal{S}$.

We also claim that such a replacement is always unambiguous.
Suppose, for the sake of contradiction, that there is some $A_p'$ such that replacement is ambiguous.
The ambiguity must be due to two tiles in the same macrotile region bonded via external strength-2 glues on \emph{different} macrosides to tiles in adjacent macrotiles, since no macroside has two strength-2 glues (see Figure~\ref{fig:macrotile}).
So there is some path in $G(A_p')$ from the external glue of one of of these tiles to the external glue to the other consisting of length-2 and length-3 subpaths through other macrotile regions, each consisting of tiles of a common macrotile.
So this path can be unambiguously replaced with a path from tiles in $\mathcal{S}$ from one side of a tile location to the other side, with some tile of $\mathcal{S}$ able to attach at this location.
But this yields a producible assembly of $\mathcal{S}$ (and thus a subassembly of $A$) with a cycle, a contradiction.
Since constructing $A_p$ from $A_p'$ is always unambiguous, and $A_p$ is a subassembly of $A$, $A_p'$ is a subassembly of $A'$.
Then by Lemma~\ref{lem:subassembly-of-producible-is-not-terminal}, $A_p'$ is not terminal.
\end{proof}

\begin{theorem}
\label{thm:mainresult}
Let $\mathcal{S} = (T, f, 1)$ be a 2HAM system with UMFTA $A$.
Then there exists a factor-2 size-separable 2HAM system $\mathcal{S}' = (T', f', 2)$ with UMFTA $A'$ and $|\mathcal{S}'| \leq 8|\mathcal{S}|$.
Furthermore, $A'$ has the shape of $A$ scaled by a factor of~2.
\end{theorem}

\begin{proof}
We modify the construction used in the proof of Lemma~\ref{lem:scaling} in two ways.
First, we apply Lemma~\ref{lem:1-occurrence-path-wherever} to create a path of 1-occurrence tiles from $e'$, the edge between a pair of 1-occurrence tiles utilized in the proof of Lemma~\ref{lem:scaling}, to the edge $e$ reached after half of a complete counterclockwise traversal of the boundary of $A$ starting at $e'$.
This is done to the original system $\mathcal{S}$, before the macrotile conversion is performed.

Next, we modify the macrotiles corresponding to 1-occurrence tiles in the unique terminal assembly $A$ of this modified system $\mathcal{S}$.
By Lemma~\ref{lem:1-occurrence-glues-unique}, the external glues used between every pair of 1-occurrence macrotiles, particularly those along the path of 1-occurrence macrotiles just constructed, appear only once in $A'$.
Recall that Lemma~\ref{lem:scaling} uses two glues to attach each pair of adjacent macrotiles.
The first glue visited during a counterclockwise traversal starting at $e'$ has strength~2, while the second has strength~1. 
We increase the strength of the second glue to 2, and eliminate the internal glue closest to this macroside in the macrotile closer to $e'$ (the macrotile further from $e'$ does not have this glue).
See Figure~\ref{fig:modified-macrotile} for an example.

\begin{figure}[ht]
\centering
\includegraphics[scale=1.0]{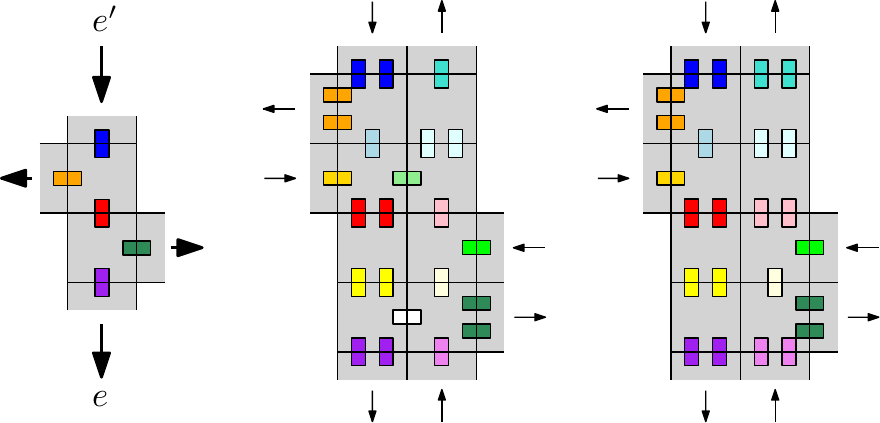}
\caption{Modifying 1-occurrence macrotiles along the path from $e'$ to $e$ to eliminate internal glues.
The original assembly with breadth-first search directions (left) yields the macrotile assembly (center) by Lemma~\ref{lem:scaling}. The modification in the proof of Theorem~\ref{thm:mainresult} (right) removes internal glues by strengthening unique external glues.}
\label{fig:modified-macrotile}
\end{figure}

We perform this modification for all of the (1-occurrence) macrotiles along the path from $e'$ to $e$.
Clearly the same terminal macrotile assembly exists as before, as we have only decreased the constraints for tiles to attach.
Moreover, since we only created pairs of strength-2 glues on macroedges with (macro)glue-pairs that appear exactly once, it is still the case that no pair of tiles can bond to external glues on the same macroside.
So the terminal assembly uniqueness argument in the proof of Lemma~\ref{lem:scaling} still holds.

This modified $A'$ has a path through the center of the macrotiles from $e'$ to $e$ containing no glues.
So $A'$ has a 2-edge cut of $G(A')$ consisting of two (non-macro)edges of $G(A)$ forming half of $e'$ and $e$, respectively. 
These two glues both have strength~2, since they external glues between 1-occurrence tiles just modified to no longer include strength-1 glues.
We reduce the strength these two glues to~1, yielding a 2-edge cut with total strength~2.

\textsc{Size-separability.} We claim that this system is factor-2 size-separable, which is true if the only 2-stable subassemblies containing tiles adjacent to both edges of the 2-edge cut have size $|A'|/2$.
Consider the 2-stability of one of these assemblies, called $A_{\rm half}'$.
Clearly tiles adjacent to the cut are at opposite ends of a counterclockwise traversal of $A_{\rm half}'$, and moreover a counterclockwise traversal (the traversal along the boundary of $A'$) visits all tiles in $A_{\rm half}'$.

Suppose some tile $t$ along the traversal is missing.
We claim that the subassembly $A_{\rm sub}$ consisting of the remainder of the tiles in the macrotile $m$ containing $t$ and any subtrees visited by the traversal before reaching the counterclockwise-most tile of $m$ (northeast in Fig.~\ref{fig:macrotile}) are only attached to the rest of the assembly by a single strength-1 glue.
Observe in Figure~\ref{fig:macrotile} that the northeast tile is attached to the next tile in the traversal (the southeast tile of the north macrotile) by a strength-1 glue.
Moreover, any path in $G(A_{\rm half}')$ from $A_{\rm sub}$ to a location not in $A_{\rm sub}$ goes through $m$.
It can be verified exhaustively that removing any tile of $m$ leaves only this single strength-1 glue connecting $A_{\rm sub}$ to the rest of $A_{\rm half}'$.

So $A_{\rm half}'$ is 2-stable, and any subassembly containing the first and last tiles visited by a counterclockwise traversal of the boundary of $A_{\rm half}'$ missing any tiles visited by the traversal (all of the tiles in $A_{\rm half}'$) is not 2-stable.
So the only 2-stable subassemblies containing tiles adjacent to both edges of the 2-edge cut have size are $A_{\rm half}'$ and the other half of the cut, both with size $|A'|/2$.

The total size of this modified macrotile system is at most $8|\mathcal{S}|$, since invoking Lemma~\ref{lem:1-occurrence-path-wherever} increases the size of the system at most a factor of~2 and so the macrotile system has size at most $4 \cdot 2|\mathcal{S}|$.
\end{proof}

\section{Open Problems}

For temperature-1 systems with mismatch-free unique terminal assemblies, our result is nearly as tight as possible.
Scaling to at least a factor of~2 and using temperature of at least~2 are both necessary, since any temperature-1 system or system with a tree-shaped assembly is at most factor-$(1 + 1/|A|)$ size-separable.
The only remaining opportunity for improvement is to reduce the number of tile types used to less than $8|\mathcal{S}|$.

We contend that our result is a first step in understanding what is possible in size-separable systems, and a large number of open problems remain.
Perhaps the most natural problem is to extend this result to the same set of systems, except permitting mismatches.
We conjecture that a similar result is possible there:

\begin{conjecture}
Let $\mathcal{S} = (T, f, 1)$ be a 2HAM system with unique finite terminal assembly $A$.
Then there exists factor-2 size-separable system $\mathcal{S}' = (T', f', 2)$ with a unique finite terminal assembly $A'$ and $|\mathcal{S}'| = O(\mathcal{S})$. 
Furthermore, $A'$ has the shape $A$ scaled by a factor of $O(1)$. 
\end{conjecture}

Extending the result to mismatch-free systems at higher temperatures also is of interest because these systems are generally capable of much more efficient assembly.
Soloveichik and Winfree~\cite{Soloveichik-2007a} prove that one can construct a temperature-2 system that uses an optimal number of tiles (within a constant factor) to construct any shape, provided one is allowed to scale the shape by an arbitrary amount, and it is likely their construction can be modified to be factor-2 size-separable.
However, it remains open to achieve high-factor size-separable systems at temperature~2 using only a small scale factor. 

\begin{conjecture}
Let $\mathcal{S} = (T, f, 1)$ be a 2HAM system with UMFTA $A$.
Then there exists factor-2 size-separable system $\mathcal{S}' = (T', f', 2)$ with a unique terminal assembly $A'$ and $|\mathcal{S}'| = O(\mathcal{S})$. 
Furthermore, $A'$ has the shape $A$ scaled by a factor of $O(1)$. 
\end{conjecture}

In the interest of applying size-separability to system in the staged model of tile self-assembly, we pose the problem of developing size-separable systems with multiple terminal assemblies.
Of course, one can construct systems where the smallest terminal assembly is less than half the size of the largest terminal assembly, ensuring that the system cannot even be factor-1 size-separable.
But given a system whose ratio of smallest to largest terminal assembly is $c$, is a size-separable system with optimal factor~$\frac{2}{c}$ always possible?

\begin{conjecture}
Let $\mathcal{S} = (T, f, 1)$ be a 2HAM system with finite terminal assemblies $A_1, A_2, \dots, A_k$ with $A_1$ and $A_k$ the smallest and largest terminal assemblies.
Then there exists factor-$|A_k|/|A_1|$ size-separable system $\mathcal{S}' = (T', f', 2)$ with $|\mathcal{S}'| = O(\mathcal{S})$ and mismatch-free terminal assemblies $A_1', A_2', \dots, A_k'$ where $A_i'$ has the shape of $A_i$ scaled by a factor of $O(1)$.
\end{conjecture}

We close by conjecturing that not every system can be made size-separable by paying only a constant factor in scale and tile types.
We ask for an example of such a system:

\begin{conjecture}
There exists a 2HAM system $\mathcal{S} = (T, f, \tau)$ with a unique finite terminal assembly $A$ such that any factor-2 size-separable system $\mathcal{S}' = (T', f', \tau')$ with unique finite terminal assembly $A'$ with the shape of $A$ either has $|\mathcal{S}'| \geq 100|\mathcal{S}|$ or the scale of $A'$ is at least 100. 
\end{conjecture}

\section*{Acknowledgments}

We thank the anonymous UCNC reviews for their comments that greatly improved the presentation of the paper.
 
\bibliographystyle{abbrv}
\bibliography{size_separable}

\end{document}